\documentclass{sig-alternate}
\makeatletter
\def\doi#1{\gdef\@doi{#1}}\def\@doi{}
\global\copyrightetc{Copyright \the\copyrtyr\ 
ACM \the\acmcopyr\ ...\$15.00}
\toappear{\the\boilerplate\par{\confname{\the\conf}} \the\confinfo\par \the\copyrightetc.\ifx\@doi\@empty\else\par\@doi.\fi}
\newfont{\mycrnotice}{ptmr8t at 7pt}
\newfont{\myconfname}{ptmri8t at 7pt}
\let\crnotice\mycrnotice%
\let\confname\myconfname%
\makeatother

\usepackage[latin1]{inputenc}

\usepackage{color}
\usepackage[plainpages=false,pdfpagelabels,colorlinks=true,citecolor=blue,hypertexnames=false]{hyperref}
\usepackage{enumerate}
\usepackage{comment}
\usepackage{mathrsfs}
\usepackage{bm}
\usepackage{stmaryrd}
\usepackage[all]{xypic}
\usepackage{array}

\DeclareBoldMathCommand{\bbeta}{\beta}

\begin{document}

\definecolor{bleu}{rgb}{0,0,0.8}
\def\todo#1{{\color{red} #1}}
\def\gathen#1{{#1}}

\newtheorem{theo}{Theorem}[section]
\newtheorem{lem}[theo]{Lemma}
\newtheorem{prop}[theo]{Proposition}
\newtheorem{cor}[theo]{Corollary}
\newtheorem{quest}[theo]{Question}
\newtheorem{ex}[theo]{Example}
\newtheorem{rem}[theo]{Remark}
\newtheorem{deftn}[theo]{Definition}

\newcommand{\N}{\mathbb N}
\newcommand{\F}{\mathbb F}
\newcommand{\Z}{\mathbb Z}
\newcommand{\Chi}{\Xi}

\newcommand{\softO}{O\tilde{~}}
\newcommand{\mA}{\mathbf A}
\newcommand{\mB}{\mathbf B}
\newcommand{\mC}{\mathbf C}
\newcommand{\mD}{\mathbf D}
\newcommand{\pcurv}{\mA_p}
\newcommand{\norm}{\mathcal{N}}
\newcommand{\GL}{\text{\rm GL}}
\newcommand{\SL}{\text{\rm SL}}

\newcommand{\cen}{\mathcal Z}
\newcommand{\D}{\mathcal D}

\newcommand{\M}{\mathcal M}
\newcommand{\Mpol}{\mathsf M}

\newcommand{\trp}{{}^{\text t}}

\permission{Permission to make digital or hard copies of all or part of this work for personal or classroom use is granted without fee provided that copies are not made or distributed for profit or commercial advantage and that copies bear this notice and the full citation on the first page. Copyrights for components of this work owned by others than ACM must be honored. Abstracting with credit is permitted. To copy otherwise, or republish, to post on servers or to redistribute to lists, requires prior specific permission and/or a fee. Request permissions from Permissions@acm.org. }

\conferenceinfo{ISSAC '14,}{July 21 - 25, 2014, Kobe, Japan.}
\CopyrightYear{2014}
\crdata{978-1-4503-2501-1/14/07}
\doi{http://dx.doi.org/10.1145/2608628.2608650}


\title{A fast algorithm for computing\\
the characteristic polynomial of the p-curvature}
%
%
%
%
%

\numberofauthors{3} 
\author{
%
%
\alignauthor Alin Bostan\\
\affaddr{INRIA (France)}\\
  \email{\normalsize \textsf{alin.bostan@inria.fr}}
\alignauthor Xavier Caruso\\
  \affaddr{Universit\'e Rennes 1}\\
  \email{\normalsize \textsf{xavier.caruso@normalesup.org}}
\alignauthor \'Eric Schost\\
  \affaddr{Western University}\\
  \email{\normalsize \textsf{eschost@uwo.ca}}
}

\date\today

\maketitle

\begin{abstract} 
  We discuss theoretical and algorithmic questions related to the
  $p$-curvature of differential operators in characteristic $p$. Given
  such an operator~$L$, and denoting by $\Chi(L)$ the characteristic
  polynomial of its $p$-curvature, we first prove a new, alternative,
  description of $\Chi(L)$. This description turns out to be
  particularly well suited to the fast computation of $\Chi(L)$ when
  $p$ is large: based on it, we design a new algorithm for computing
  $\Chi(L)$, whose cost with respect to~$p$ is $\softO(p^{0.5})$
  operations in the ground field. This is remarkable since, prior to
  this work, the fastest algorithms for this task, and even for the
  subtask of deciding nilpotency of the $p$-curvature, had merely
  slightly subquadratic complexity $\softO(p^{1.79})$.
\end{abstract}

\vspace{1mm}
 \noindent
 {\bf Categories and Subject Descriptors:} \\
\noindent I.1.2 [{\bf Computing Methodologies}]:{~} Symbolic and Algebraic
  Manipulation -- \emph{Algebraic Algorithms}

 \vspace{1mm}
 \noindent
 {\bf General Terms:} Algorithms, Theory

 \vspace{1mm}
 \noindent
 {\bf Keywords:} Algorithms, complexity, differential equations, $p$-curvature.

\medskip
\section{Introduction}\label{sec:intro}

This article deals with some algorithmic questions related to linear
differential operators in positive characteristic $p$. More precisely, we
address the problem of the efficient computation of the characteristic
polynomial of the $p$-curvature of such a differential operator $L$. Roughly
speaking, the $p$-curvature of $L$ is a matrix that measures
to what extent the solution space of $L$ has dimension close to its order. The
theory was initiated in the 1970s by Katz, Dwork and
Honda~\cite{Katz70,Dwork82,Honda81} in connection with one of Grothendieck's
conjectures which states that an irreducible linear differential operator with
coefficients in $\mathbb{Q}(x)$ admits a basis of algebraic solutions over
$\mathbb{Q}(x)$ if and only if its reductions modulo $p$ admit a zero
$p$-curvature for almost all primes $p$. 

Let $k$ be \emph{any} field of characteristic~$p$, and let $k(x)\langle \partial
\rangle$ be the algebra of differential operators with coefficients in $k(x)$,
with the commutation rule $\partial x = x\partial + 1$. The $p$-curvature of a
differential operator $L$ of order $r$ in $k(x)\langle \partial \rangle$,
hereafter denoted $\mA_p(L)$, is the $(r\times r)$ matrix with coefficients in
$k(x)$, whose $(i,j)$ entry is the coefficient of $\partial^{i}$ in the
remainder of the Euclidean (right) division of $\partial^{p+j}$ by $L$, for $0
\le i,j < r$.

We focus on the computation in good complexity, notably with respect to the
parameter~$p$, of the characteristic polynomial $\Chi(L)$ of the $p$-curvature
$\pcurv(L)$. An important sub-task is to decide efficiently whether
$\pcurv(L)$ is nilpotent. By a celebrated theorem of the
Chudnovskys'~\cite{ChCh85}, least order differential operators satisfied by
$G$-series possess reductions modulo $p$ with nilpotent $p$-curvatures for
almost all primes~$p$.

Studying the complexity of the computation of $\Chi(L)$ is an
interesting problem in its own right. This computation is for instance
one of the basic steps in algorithms for factoring linear differential
operators in
characteristic~$p$~\cite{vanDerPut95,vanDerPut96,Cluzeau03}. Additional
motivations for studying this question come from concrete
applications, in combinatorics~\cite{BoKa08b,BoKa08a} and in
statistical physics~\cite{BBHMWZ08}, where the $p$-curvature serves as
an \emph{a posteriori} certification filter for differential operators
obtained by guessing techniques from power series expansions. In such
applications, the prime number $p$ may be quite large (thousands, or
tens of thousands), since its value is lower bounded by the precision
of the power series needed by guessing, which is typically large for
operators of large size. This explains our choice of considering $p$
as the most important complexity parameter.

\smallskip\noindent{\bf Previous work.} Since $k(x)\langle \partial
\rangle$ is noncommutative, binary powering cannot be used to compute
$\partial^p \bmod L$.  Katz~\cite{Katz82} gave the first algorithm for
$\mA_p(L)$, based on the recurrence $\mA_1 = \mA,\quad \mA_{k+1} =
\mA_k' + \mA \mA_k,$ where $\mA \in \mathscr{M}_r(k(x))$ is the
companion matrix associated to $L$.  This algorithm, as well as its
variants~\cite[\S13.2.2]{PuSi03} and~\cite[Prop.~3.2]{Cluzeau03} have
complexity quadratic in~$p$.  The first subquadratic algorithm was
designed in~~\cite[\S6.3]{BoSc09}.  It has complexity
$\softO(p^{1.79})$ and it is based on the observation that the
$p$-curvature $\pcurv(L)$ is obtained by applying the matrix operator
$(\partial+\mA)^{p-1}$ to~$\mA$, and on a baby steps/giant steps
algorithm for applying differential operators to polynomials.

Several partial results concerning the $p$-curvature were obtained
in~\cite{BoSc09}: computation of $\pcurv(L)$ in {$O(\log(p))$} for \emph{first
order} operators and in quasi-linear time {$\softO(p)$} for \emph{certain
second order} operators; algorithms of complexity {$\softO({p}^{0.5})$} for
deciding nilpotency of $\mA_p(L)$ for \emph{second order} operators, and {$\softO(p)$} for the
nullity of $\mA_p(L)$ for \emph{arbitrary operators}.

\smallskip\noindent{\bf Our contribution.} 
Prior to this work, the computation of the characteristic polynomial of 
the $p$-curvature required the computation of the $p$-curvature itself 
as a preliminary step. We manage to compute $\Chi(L)$ without 
$\pcurv(L)$ by exploiting in a completely explicit and elementary way 
the fact that the Weyl algebra $k[x]\langle \partial \rangle$ is a 
central separable (Azumaya) algebra over its centre $k[x^p,\partial^p] 
$, and thus endowed with a \emph{reduced norm 
map}~\cite{Revoy73,KnOj74,CaLeBo12}.

Our crucial observation is that the characteristic polynomials of the 
$p$-curvature of elements in $k(x)\langle \partial \rangle$ are closely 
related to other polynomials associated to operators 
lying in the skew ring $k(\theta)\langle \partial^{\pm 1} \rangle$ on 
which the multiplication is determined by the rule $\partial \theta = 
(\theta + 1) \partial$. More precisely, given such an operator $L$, we
define its $p$-curvature $\mB_p(L)$ and compare its characteristic
polynomial to that of $\mA_p(L)$ when $L$ makes sense in both rings
$k(x)\langle \partial \rangle$ and $k(\theta)\langle \partial^{\pm 1} 
\rangle$ (Theorem~\ref{theo:comparison}). In addition, the computation
of the characteristic polynomial of $\mB_p(L)$ reduces to that of a 
matrix factorial of length $p$, which can be performed in 
$\softO(p^{0.5})$ operations in~$k$ via the baby steps/giant steps 
approach in~\cite{ChCh88}. This allows us to compute $\Chi(L)$ in 
complexity quasi-linear in $p^{0.5}$.

\smallskip\noindent{\bf Structure of the paper.}  In Section
\ref{sec:rings}, we introduce all rings of differential operators that
we need and recall their basic properties. Section \ref{sec:pcurv} is
devoted to the theoretical study of the $p$-curvature of there
differential operators and culminates in the proof of
Theorem~\ref{theo:comparison}. In Section \ref{sec:algo}, we move to
applications to algorithmics: after some preliminaries, we describe
our main algorithm for computing $\Chi(L)$ in complexity
$\softO(p^{0.5})$. We conclude with the implementation of our
algorithm and some benchmarks and applications.

\smallskip\noindent{\bf Acknowledgements.} 
We would like to thank the referees
for their insightful remarks. This work was supported by NSERC,
the CRC program and the MSR-Inria Joint Centre.


\section{Differential operators}
\label{sec:rings}

Throughout this article, $p$ is a prime number and the letter $k$ 
denotes a field of characteristic $p$. We use the classical notations 
$k[x]$ and $k(x)$ to refer to the ring of polynomials over $k$ and the 
field of rational fractions over $k$ respectively. We recall that $k(x)$ 
is the field of fractions of $k[x]$.


\subsection{Usual differential operators}

The ring of differential operators over $k(x)$, that we shall denote
$k(x)\langle\partial\rangle$ in the sequel, is a noncommutative ring 
whose elements are polynomials in $\partial$ of the form:
$$L=f_0(x) + f_1(x) \partial + f_2(x) \partial^2 + \cdots + f_r(x)
\partial^r$$ where $f_i(x)$ are elements in $k(x)$. The multiplication 
in $k(x)\langle\partial\rangle$ is determined by the so-called Leibniz 
rule:
\begin{equation}
\label{eq:leibniz}
\partial  f = f \partial + f'
\end{equation}
where $f$ is in $k(x)$ and $f'$ denotes its derivative.
Recall~\cite{ore33} that $k(x)\langle\partial\rangle$ is a
noncommutative Euclidean ring (on the left and on the right); this
implies that $k(x)\langle\partial\rangle$ is principal and that there
is a notion of left and right gcd's over this ring. Euclid's algorithm
and B\'ezout's theorem extend as well.

For the purpose of this article, we shall need to invert
formally the variable $\partial$. To do this, we consider the additive 
group consisting of Laurent polynomials in $\partial$ over $k(x)$, 
\emph{i.e.} polynomials having the form:
$$f_{-s}(x) \partial^{-s} + \cdots + f_0(x) + \cdots + f_r(x)
\partial^r \quad \text{(with } s, r \in \N \text{)}$$
and define a multiplication on it by letting for all $f \in k(x)$:
\begin{equation}
\label{eq:iteribp}
\partial^{-1} f = \sum_{i=0}^{p-1} (-1)^i f^{(i)} \: \partial^{-i-1},
\end{equation}
where $f^{(i)}$ denotes the $i$-th derivative of $f$. The latter formula 
is obtained by performing $p-1$ integrations by parts and noting that the 
$p$-th derivative of any element in $k(x)$ vanishes. It is an exercise to 
check that Eq.~\eqref{eq:iteribp} defines a ring structure on 
$k(x)\langle\partial^{\pm 1}\rangle$, which extends the one of
$k(x)\langle\partial\rangle$.

We will often work with the sets $k[x]\langle \partial \rangle$ and
$k[x]\langle \partial^{\pm 1} \rangle$ consisting of all operators in
$k(x)\langle \partial \rangle$ and $k(x)\langle \partial^{\pm 
1}\rangle$ respectively, whose coefficients belong to $k[x]$.  It is 
easily seen from \eqref{eq:leibniz} and \eqref{eq:iteribp} that
$k[x]\langle\partial\rangle$ is actually a subring of
$k(x)\langle\partial\rangle$ and that $k[x]\langle\partial^{\pm
1}\rangle$ is a subring of $k(x)\langle \partial^{\pm 1}\rangle$.

Recall that the \emph{centre\/} of a noncommutative ring $A$ is the 
subring of $A$ consisting of all elements which commute with all 
elements in $A$. The centres of $k(x)\langle\partial\rangle$ and 
$k(x)\langle\partial^{\pm 1}\rangle$ are $k(x^p)[\partial^p]$ and 
$k(x^p)[\partial^{\pm p}]$, respectively; the same holds for their 
counterparts with polynomial coefficients~\cite{Revoy73,vanDerPut95}. 
They will play an important role in this article.


\subsection{The Euler operator}\label{ssec:euler}

The Euler operator is the element $x \partial$. One important feature of 
it is that it satisfies simple relations of commutation against $x$ and 
$\partial$, namely:
$$x \cdot (x \partial) = (x \partial - 1) \cdot x \quad \text{and}
\quad \partial \cdot (x\partial) = (x \partial + 1) \cdot \partial.$$
This motivates the following definition. We introduce a new variable
$\theta$ and consider the field $k(\theta)$ of rational fractions over
$k$ in the variable $\theta$. We define the noncommutative ring
$k(\theta)\langle\partial\rangle$
(resp.\ $k(\theta)\langle\partial^{\pm 1} \rangle$) whose elements are
polynomials (resp.\ Laurent polynomials) over $k(\theta)$ in the
variable $\partial$, and on which the multiplication follows the rule:
\begin{equation}
\label{eq:euler}
\partial^i g(\theta) = g(\theta + i) \: \partial^i, \quad \text{for all} \; i \in \Z \; \text{and} \; g \in k(\theta).
\end{equation}
Just like $k(x)\langle\partial\rangle$, the ring 
$k(\theta)\langle\partial\rangle$ is Euclidean on the left and on the
right and therefore admits left (resp.\ right) gcd's.
The centres of $k(\theta)\langle\partial\rangle$ and $k(\theta) 
\langle\partial^{\pm 1}\rangle$ are $k(\theta^p-\theta)[\partial^p]$ and 
$k(\theta^p-\theta)[\partial^{\pm p}]$, respectively.

As we did for  usual differential operators, we define
$k[\theta]\langle\partial\rangle$ and $k[\theta]\langle\partial^{\pm
1}\rangle$ as the subsets of respectively $k(\theta)
\langle\partial\rangle$ and $k(\theta) \langle\partial^{\pm 1}\rangle$
consisting of all operators having coefficients in $k[\theta]$;
Formula~\eqref{eq:euler} shows that $k[\theta]\langle\partial\rangle$
and $k[\theta]\langle\partial^{\pm 1}\rangle$ are closed under
multiplication, and hence are rings.
It is easily seen that the following two morphisms of $k$-algebras
$$\begin{array}{rcl}
k[x]\langle\partial^{\pm 1}\rangle & \rightleftarrows &
k[\theta]\langle\partial^{\pm 1}\rangle \smallskip \\
x & \mapsto & \theta \partial^{-1} \smallskip \\
x \partial & \mapsfrom & \theta \smallskip \\
\partial^{\pm 1} & \leftrightarrow & \partial^{\pm 1} 
\end{array}$$
define inverse isomorphisms between the rings $k[x]\langle\partial^{\pm 
1}\rangle$ and $k[\theta]\langle\partial^{\pm 1}\rangle$. Beware 
however that these isomorphisms \emph{do not\/} extend to isomorphisms 
between 
$k(x)\langle\partial^{\pm 1}\rangle$ and $k(\theta)\langle\partial^{\pm 
1}\rangle$. Indeed, an element of $k[x]$ (resp. of $k[\theta]$) is in
general not invertible in $k(\theta)[\partial^{\pm 1}]$ (resp. 
in $k(x)\langle\partial^{\pm 1}\rangle$).

We remark that under the above identification, the central element
$\theta^{p}-\theta$ corresponds to $x^p \partial^p$.


\section{A theoretical study of the\\ {\Large \lowercase{$p$}}-curvature}
\label{sec:pcurv}


\subsection{Definitions and first properties}
\label{ssec:defChi}


\noindent{\bf \rm Over $k(x)\langle\partial\rangle$.}
Let $L$ be a differential polynomial in
$k(x)\langle\partial\rangle$. We denote by $k(x)\langle\partial\rangle
\: L$ the set of right multiples of $L$, that is the set of
differential polynomials of the form $Q L$ for some $Q \in
k(x)\langle\partial\rangle$. Clearly, it is a vector space over
$k(x)$. The quotient $M_L = k(x)\langle\partial\rangle /
k(x)\langle\partial\rangle L$ is a finite dimensional vector space
over $k(x)$ and a basis of it is $(1, \partial, \ldots,
\partial^{r-1})$, where $r$ denotes the degree of $L$ with respect to
$\partial$.

\begin{deftn}
  The \emph{$p$-curvature} of
  $L \in k(x)\langle\partial\rangle$ is the $k(x)$-linear endomorphism of $M_L$ induced by the
  multiplication by the central element $\partial^p$.
\end{deftn}


Given $L \in k(x)\langle\partial\rangle$, we denote by $\mA_p(L)$ the 
matrix of the $p$-curvature of $L$ in the basis $(1, \partial, \ldots, 
\partial^{r-1})$ and by $\chi(\pcurv(L))$ its characteristic polynomial:
$$\chi(\pcurv(L))(X) = \det(X \cdot \text{Id} - \pcurv(L)).$$ 
It is well-known~\cite{vanDerPut96} that all coefficients of 
$\chi(\pcurv(L))$ lie in $k(x^p)$. For our purposes, it will be 
convenient to renormalize $\chi(\pcurv(L))$ as follows: we set 
$$\Chi_{x,\partial}(L) = f_r(x)^p \cdot \chi(\pcurv(L))(\partial^p)$$ 
where $f_r(x)$ is the leading coefficient of $L$. We note that 
$\Chi_{x,\partial}(L)$ belongs to $k(x^p)[\partial^p]$, \emph{i.e.} to 
the centre of $k(x) \langle\partial\rangle$.

\begin{lem}
  \label{lem:propChi}
  Let $L$ be a differential operator in $k(x)\langle \partial\rangle$.
  \begin{enumerate}[(i)]
  \item The degree of $\Chi_{x,\partial}(L)$ in the variable 
    $\partial^p$ is equal to the degree of $L$ in the variable $\partial$.
  \item $L$ divides $\Chi_{x,\partial}(L)$ on both sides.
  \item if $L$ is irreducible in $k(x)\langle \partial\rangle$, then
  $\Chi_{x,\partial}(L)$ is a power of an irreducible element of
  $k(x^p)[\partial^p]$.
  \end{enumerate}
  Besides, the map $\Chi_{x,\partial}$ is multiplicative.
\end{lem}
\begin{proof}
  The first assertion is obvious, while the second one 
  is a direct consequence of Cayley-Hamilton Theorem.


  We are going to prove \emph{(iii)} by contradiction: we pick  an irreducible differential operator $L \in k(x)
  \langle\partial\rangle$ and 
  assume that there exist two distinct irreducible polynomials 
  $N_1$ and $N_2$ that both divide $\chi(\mA_p(L))$. Since these 
  polynomials are coprime, there must exist $i \in \{1, 2\}$ such
  that $N_i(\partial^p)$ is coprime with $L$. By B\'ezout's theorem,
  this implies that $N_i(\partial^p)$ defines an invertible 
  endomorphism of $M_L$. This contradicts the fact that $N_1$
  divides the characteristic polynomial of $\partial^p$ acting on
  this space.

  The Leibniz rule \eqref{eq:leibniz}
  implies that the leading coefficient of the product $L_1 L_2$ is
  equal to the product of the leading coefficients of the
  factors. Moreover, by \cite[Lemma 1.13]{ClvH04}, we know that $\chi \circ \pcurv$
  is multiplicative. The multiplicativity of $\Chi_{x,\partial}$
  follows.
\end{proof}

The multiplicativity property allows us to extend the map $\Chi_{x,
  \partial}$ to $k(x) \langle\partial^{\pm 1}\rangle$. Indeed given a
differential operator $L$ in the latter ring, there exists an integer
$n$ such that $L \cdot \partial^n$ lies in $k(x)
\langle\partial\rangle$ and we can define
$\Chi_{x, \partial}(L) = \partial^{-pn} \cdot \Chi_{x, \partial}(L
\cdot \partial^n)$. The multiplicativity property and the fact that
$\Chi_{x,\partial}(\partial) = \partial^p$ show that this definition does not
depend on the choice of~$n$. The extended map $\Chi_{x, \partial}$
takes its values in $k(x^p)[\partial^{\pm p}]$, that is again the
centre of $k(x) \langle\partial^{\pm 1}\rangle$.

\smallskip\noindent{\bf \rm Over $k(\theta)\langle\partial\rangle$.}
Following \cite[\S 5]{vdPSi97}, we extend the definition of 
$p$-curvature
to differential operators over $k(\theta)$.

Given an element $L$ in $k(\theta)\langle\partial\rangle$, we consider the 
quotient $k(\theta)\langle\partial\rangle / k(\theta)\langle\partial\rangle 
L$ and define the \emph{$p$-curvature} of $L$ as the endomorphism of 
this space given by multiplication by $\partial^p$.
As before, the quotient above is a finite dimensional vector space over 
$k(\theta)$ and admits $(1, \partial, \ldots, \partial^{r-1})$ as a basis, where $r$ 
denotes the degree of $L$ with respect to $\partial$.
Let us denote by $\mB_p(L)$ the matrix of the $p$-curvature of $L$ in the 
basis $(1, \partial, \ldots, \partial^{r-1})$ considered above. 
The following easy lemma gives an explicit formula for it.

\begin{lem}
  \label{lem:pcurvtheta}
  Let $L \in k(\theta)\langle\partial\rangle$ be a differential
  polynomial of degree $r$ with respect to the variable $\partial$.
  Let $\mB(\theta) \in \mathscr{M}_r(k(\theta))$ denote the companion
  matrix of $L$.  Then:
  $$\mB_p(L) = \mB(\theta) \cdot \mB(\theta+1) \cdots \mB(\theta + p - 1).$$
\end{lem}


\begin{rem}
  It may happen that the same differential operator $L$ makes sense in
  both rings $k(x)\langle\partial\rangle$ and
  $k(\theta)\langle\partial \rangle$.  In that case, one should be
  very careful that the $p$-curvature computed in $k(x)\langle\partial\rangle$ has in
  general nothing to do with the $p$-curvature computed in $k(\theta)\langle\partial \rangle$. 
  For instance, they might have different sizes.  If confusion
  may arise, we shall speak about ``$p$-curvature with respect to
  $(x,\partial)$'' and ``$p$-curvature with respect to
  $(\theta,\partial)$'' respectively.
\end{rem}

Keeping our $L$ in $k(\theta)\langle\partial\rangle$, we set:
$$\Chi_{\theta,\partial}(L) = g_r(\theta) \cdot g_r(\theta+1) \cdots
g_r(\theta + p - 1) \cdot \chi(\mB_p(L))(\partial^p),$$ where $\chi$
refers to the characteristic polynomial and $g_r(\theta)$ denotes the
leading coefficient of $L$. The three properties of Lemma
\ref{lem:propChi} extend readily to this new setting. Using
multiplicativity, the function $\Chi_{\theta,\partial}$ can be
extended to $k(\theta)\langle \partial^{\pm 1} \rangle$.

\begin{lem}\label{lemma:central}
  The function $\Chi_{\theta,\partial}$ takes its values in the centre
  of $k(\theta)\langle\partial^{\pm 1}\rangle$, that is $k(\theta^p -
  \theta)[\partial^{\pm p}]$.
\end{lem}

\begin{proof}
  Pick some $L \in k(\theta)\langle\partial\rangle$ and denote by 
  $g_r(\theta)$ its leading coefficient. Clearly, the product 
  $$g_r(\theta) \cdot g_r(\theta+1) \cdots g_r(\theta + p - 1)$$ is
  invariant under the substitution $\theta \mapsto \theta + 1$ and
  thus can be written as a rational fraction in $\theta^p - \theta$.
  In the same way, the matrix $\mB_p(L)$ is similar to the same matrix
  where we have made the substitution $\theta \mapsto \theta +
  1$. This implies that all the coefficients of $\chi(\mB_p(L))$ are
  invariant under $\theta \mapsto \theta + 1$. Therefore as before,
  they are rational fractions in $\theta^p - \theta$. 
\end{proof}


\subsection{A comparison theorem}

The aim of this section is to show that the two maps
$\Chi_{x,\partial}$ and $\Chi_{\theta,\partial}$ defined above
coincide on $k[x]\langle\partial^{\pm 1}\rangle \simeq
k[\theta]\langle\partial^{\pm 1}\rangle$.

\smallskip\noindent{\bf \rm Comparison with a matrix algebra.}
In order to simplify notations, we will use the letter~$\D$ to
denote the ring $k[\theta]\langle\partial^{\pm 1}\rangle$. The centre
of $\D$ is $k[\theta^p - \theta][\partial^{\pm p}]$; we denote it by
$\cen$. We consider the ring extension $\cen[T]$ where $T$ is a new
variable satisfying the equation $T^p - T = \theta^p - \theta$. In a
slight abuse of notation, we shall write $\D[T]$ for $\cen[T]
\otimes_\cen \D$. We emphasize that by definition, the adjoined element 
$T$ lies in the centre of $\D[T]$.
We endow $\cen[T]$ and $\D[T]$ with an action of the cyclic additive
group $\F_p$ by letting $a$ act on $T$ as $T+a$ (and acting
trivially on $\D$). It is easily seen that the set of fixed points of
$\cen[T]$ (resp.\ $\D[T]$) under the above action is~$\cen$
(resp.~$\D$). We introduce the two matrices over $\cen[T]$:
$$\M(\theta) = \left ( 
\begin{smallmatrix} 
T \\ & T+1 \vspace{-0.5em} \\ & & \ddots \\ &&& T+p-1
\end{smallmatrix} \right ) \;\; \text{~and~} \;\;
\M(\partial) =\left ( 
\begin{smallmatrix} 
& 1  \vspace{-0.5em} \\
&& \ddots\\
&&& 1\\
\partial^p
\end{smallmatrix}
\right ).$$
We check that $\M(\partial)\M(\theta) = (\M(\theta) + 1) \M(\partial)$. 
As a consequence $\M$ uniquely entends to a ring morphism $\M : \D[T] \to 
\mathscr{M}_p(\cen[T])$. Moreover if $L$ lies in $\D[T]$ and is written 
as
$$L = \sum_{0 \leq i,j < p} a_{i,j} \theta^i \partial^j
\quad \text{with } a_{ij} \in \cen[T]$$
a closed formula for $\M(L)$ exists: it is the $(p \times p)$ matrix 
whose $(i',j')$ entry ($0 \leq i',j' < p)$ is
\begin{equation}
  \label{eq:defM}
  \M(L)_{i',j'} = \partial^{i' - j' + r} \cdot \sum_{i = 0}^{p-1}
  a_{i,r} \cdot (T+j')^i
\end{equation}
where $r$ denotes the remainder in the Euclidean division of $j' - 
i'$ by $p$. 

\begin{prop}
  \label{prop:azumaya}
  The map $\M : \D[T] \to \mathscr{M}_p(\cen[T])$ is an isomorphism of $\cen[T]$-algebras.
\end{prop}

\begin{proof}
  It is an exercise to check that $\M$ maps any element $a \in
  \cen[T]$ to $a \cdot \text{Id}$. Thanks to Eq.~\eqref{eq:defM},
  in order to prove that it is an isomorphism, we need
  to check that, knowing all $\M(L)_{i',j'}$'s (with $i'$ and $j'$
  varying in $\{0, \ldots, p-1\}$), one can recover uniquely all
  $a_{i,j}$'s (with again $i$ and $j$ varying in $\{0, \ldots,
  p-1\}$). From \eqref{eq:defM}, we see that, for any $r$, the $p$
  values $a_{i,r}$ satisfy a Vandermonde system with coefficients in
  $\cen[T]$ (recall that $\partial^p$ is invertible in this ring)
  whose determinant is:
  $$\prod_{0 \leq a < b < p} \big((T+a)-(T+b)\big)
  = \prod_{0 \leq a < b < p} \big(a-b)$$
  and hence belongs to $\F_p^\star$. Therefore they can be recovered 
  uniquely from the $\M(L)_{i',j'}$'s.
\end{proof}

\begin{cor}
\label{cor:azumaya}
The map $\M$ induces the following identifications:
{\small \begin{align*}
    k[\theta^p{-}\theta][\partial^{\pm p}][T] 
    \otimes_{k[\!\theta^p{-}\theta][\partial^{\pm p}]} 
    k[\theta]\langle \partial^{\pm 1} \rangle 
    &\! \simeq\!
    \mathscr{M}_p(k[\theta^p{-}\theta][\partial^{\pm p}][T]) \\
    k(\theta^p{-}\theta)[\partial^{\pm p}][T] \!
    \otimes_{k(\!\theta^p{-}\theta)[\partial^{\pm p}]} \!
    k(\theta)\langle \partial^{\pm 1} \rangle 
    &\! \simeq\!
    \mathscr{M}_p(k(\theta^p{-}\theta)[\partial^{\pm p}][T]) \\
    k[x^p][\partial^{\pm p}][T] \otimes_{k[x^p][\partial^{\pm p}]} 
    k[x]\langle \partial^{\pm 1} \rangle 
    &\! \simeq\!
    \mathscr{M}_p(k[x^p][\partial^{\pm p}][T]) \\
    k(x^p)[\partial^{\pm p}][T] \otimes_{k(x^p)[\partial^{\pm p}]} 
    k(x)\langle \partial^{\pm 1} \rangle 
    &\! \simeq\!
    \mathscr{M}_p(k(x^p)[\partial^{\pm p}][T])
\end{align*}}%
where, in the last two cases, $T$ satisfies $T^p - T = x^p
\partial^p$.
\end{cor}

\begin{proof}
  The first isomorphism is Proposition~\ref{prop:azumaya}; the second
  one follows by extending scalars from $k[\theta^p-\theta]$ to
  $k(\theta^p-\theta)$.  The third one follows from the identification
  $k[x]\langle\partial^ {\pm 1}\rangle \simeq
  k[\theta]\langle\partial^{\pm 1}\rangle$ which also identifies the
  centres $k[x^p][\partial^{\pm p}]$ and
  $k[\theta^p{-}\theta][\partial^{\pm p}]$.  The last isomorphism
  follows from the third one by extending scalars from $k[x^p]$ to
  $k(x^p)$.
\end{proof}

\noindent{\bf \rm The map $\Chi_{\theta,\partial}$ as a determinant.}
Let us recall that in \S \ref{ssec:defChi} we have defined a map:
$$\Chi_{\theta,\partial} : k(\theta)\langle\partial^{\pm 1}\rangle
\to k(\theta^p - \theta)[\partial^{\pm p}].$$
Using Corollary \ref{cor:azumaya}, one can 
define another map having the same domain and codomain, as follows. 
We denote by
$$\begin{array}{ll} \norm : &
k(\theta^p{-}\theta)[\partial^{\pm p}][T]
\otimes_{k(\!\theta^p{-}\theta)[\partial^{\pm p}]}
k(\theta)\langle \partial^{\pm 1} \rangle \smallskip \\
& \hspace{9em} \longrightarrow
k(\theta^p{-}\theta)[\partial^{\pm p}][T]
\end{array}$$
the map obtained by composing the second isomorphism of Corollary 
\ref{cor:azumaya} with the determinant map. 

\begin{lem}
  $\norm$ commutes with the action of $\F_p$.
\end{lem}
\begin{proof}
  Let $\sigma$ denote the mapping defined on $\D[T]$ by the identity
  on $\D$ and $T \mapsto T+1$; we extend it to $\mathscr{M}_p(\cen[T])$
  componentwise. It is enough to prove that for $L$ in $\D[T]$,
  $\norm(\sigma(L))=\sigma(\norm(L))$, since then it suffices to
  extend scalars to $k(\theta^p-\theta)$ to conclude. We are going to
  prove that for any such $L$, the equality $\M(\sigma(L)) =
  \M(\partial)^{-1} \sigma(\M(L)) \M(\partial)$ holds. Once this is
  established, taking determinants proves our claim.
  Since both mappings above are ring morphisms, it is enough to prove
  that they coincide for $L=a \in \cen[T]$, $L=\theta$ and
  $L=\partial$. In the first case, $\M(L) =a \cdot \text{Id}$ and
  $\M(\sigma(L))=\sigma(a)\cdot \text{Id}$, so the claim holds. The
  other cases follow by inspection.
\end{proof}

In particular, $\norm$ induces a map $k(\theta)\langle \partial^{\pm
  1} \rangle \to k(\theta^p{-}\theta)[\partial^{\pm p}]$ that, in a
slight abuse of notation, we still denote~$\norm$.  It is the
so-called \emph{reduced norm} map.

\begin{lem}
  \label{lem:norm}
  Let $L$ be in $k(\theta)\langle \partial^{\pm 1} \rangle$.
  \begin{enumerate}[(i)]
  \item If $L$ is in $k[\theta]\langle \partial \rangle$ of
    degree $r$ in $\partial$ and with
    coefficients of degree at most $d$ in $\theta$, then $\norm(L)$ is in
    $k[\theta^p-\theta][\partial^p]$ and has degree at most $d$ in
    $\theta^p - \theta$ and exactly $r$ in $\partial^p$.
  \item If $L$ lies in the centre $\cen$, then $\norm(L) = L^p$.
  \item If $L$ is irreducible in $k(\theta)\langle 
    \partial^{\pm 1}\rangle$, then $\norm(L)$ is a power of an irreducible element 
    of $\cen$.
  \end{enumerate}
  Besides, the map $\norm$ is multiplicative.
\end{lem}

\begin{proof}
  Suppose first that $L$ is in $k[\theta]\langle \partial
  \rangle$.  In view of the shape of $\M(\theta)$ and $\M(\partial)$, it
  is clear that $\norm(L)$ involves no negative power in $\partial^p$. 
  Moreover, we see that
  $\M(\theta^i)$ can be written as a matrix with entries of degree $i$
  in $T$; if all coefficients of $L$ have degree at most $d$ in
  $\theta$, this implies that $\norm(L)$ can be written with
  coefficients of degree at most $dp$ in $T$. Since we know that 
  $\norm(L)$ lies in $k[\theta^p-\theta][ \partial ]$ and that $T$ 
  satisfies $T^p - T = \theta^p - \theta$, we find that $\norm(L)$ has 
  degree at most $d$ in $\theta^p - \theta$ as claimed. The rest of
  \emph{(i)} follows similarly.

  We have already seen that if $L \in \cen$, then $\M(L) = L \cdot 
  \text{Id}$ and therefore $\norm(L) = L^p$.


  We prove \emph{(iii)}. To simplify notation, set $\D' = k(\theta)\langle
  \partial^{\pm 1}\rangle$ and $\cen' =
  k(\theta^p-\theta)[\partial^{\pm p}]$. Let $L$ be an irreducible
  element of $\D'$. We assume by contradiction that there exist two
  distinct irreducible polynomials $N_1, N_2 \in \cen'$ that divide
  $\norm(L)$.  Then $N_1$ and $N_2$ are coprime in $\D'$. Thus one of
  these polynomials, say $N_1$, is coprime with $L$. By B\'ezout's
  Theorem, there exists $Q \in \D'$ such that $Q L \equiv 1
  \pmod{N_1}$. Thus the image of $L$ in $\D'[T] / N_1 \D'[T]$ is
  invertible in this ring. This implies that the image of $L$ in
  $\mathscr{M}_p(\cen'[T] / N_1 \cen'[T])$ (by the isomorphism of
  Corollary \ref{cor:azumaya} composed with the canonical projection)
  is invertible as well. Therefore $\norm(L)$ has to be invertible in
  $\cen'[T] / N_1 \cen'[T]$. But, on the other hand, we had assumed
  that $N_1$ divides $\norm(L)$. This is a contradiction.

  The multiplicativity of $\norm$ follows immediately from the 
  multiplicativity of the determinant.
\end{proof}

\begin{prop}
  \label{prop:norm}
  The two maps $\Chi_{\theta, \partial}$ and $\norm$ agree.
\end{prop}

\begin{proof}
  Using multiplicativity and remarking that $\Chi_{\theta, \partial}$ 
  and $\norm$ both map $g(\theta) \in k[\theta]$ to $g(\theta) \cdot 
  g(\theta+1) \cdots g(\theta+p-1)$, we are reduced to prove that 
  $\Chi_{\theta, \partial}(L) = \norm(L)$ for any \emph{monic 
  irreducible} differential polynomial $L \in k(\theta)\langle 
  \partial \rangle$.

  Take such an $L$. Since $L$ divides $\Chi_{\theta, 
  \partial}(L)$, we can write 
  \begin{equation}
  \label{eq:factChiL}
  L \cdot L_1 L_2 \cdots L_s = \Chi_{\theta, \partial}(L)
  \end{equation}
  where $L_i$'s are monic irreducible differential operators. Set $L_0
  = L$. For $i \in \{0, \ldots, s\}$, we know that $\Chi_{\theta,
    \partial}(L_i) = N_i^{n_i}$ and $\norm(L_i) = M_i^{m_i}
  \partial^{pm'_i}$ where $N_i$ and $M_i$ are monic irreducible
  polynomials in $k(\theta^p - \theta)[\partial^p]$ and $n_i$, $m_i$
  and $m'_i$ are nonnegative integers with $n_i > 0$. Applying $\norm$
  to \eqref{eq:factChiL} gives
  \begin{equation}
  \label{eq:normChiL}
  \partial^{p m'} \cdot \prod_{i=0}^s M_i^{m_i} = N_0^{p n_0}
  \end{equation}
  where $m' = \sum_{i=0}^s m'_i$. Hence we can assume that $M_i =
  N_0$ for all $i$. Now, if $N_0 = \partial^p$, both $\Chi_{\theta, 
  \partial}(L)$ and $\norm(L)$ are powers of $\partial^p$ and we get
  the desired result by comparing degrees. On the contrary, if $N_0$
  is not $\partial^p$, Eq.~\eqref{eq:normChiL} implies that $m' = 0$ 
  and then that $m'_0 = 0$ as well. Thus $\Chi_{\theta, \partial}(L)$
  and $\norm(L)$ are both powers of $N_0$. Since they are monic and
  share the same degree, they need to be equal.
\end{proof}

\noindent{\bf\rm Consequences.}
We are now in position to prove the following theorem that compares
the maps $\Chi_{x,\partial}$ and $\Chi_{\theta,\partial}$.

\begin{theo}
  \label{theo:comparison}
  The following diagram commutes:
  $$\xymatrix @C=35pt @R=20pt {
    k[\theta]\langle\partial^{\pm 1}\rangle 
    \ar[r]^-{\Chi_{\theta, \partial}} 
    \ar[d]^-{\sim}_-{\theta \mapsto x \partial} &
    k[\theta^p - \theta][\partial^{\pm p}] 
    \ar[d]_-{\sim}^-{\theta^p - \theta \mapsto x^p \partial^p} \\
    k[x]\langle\partial^{\pm 1}\rangle \ar[r]^-{\Chi_{x, \partial}} &
    k[x^p][\partial^{\pm p}] 
  }$$
\end{theo}

\begin{proof}
  By Proposition \ref{prop:norm}, we know that the image of an element
  $L \in k[\theta]\langle\partial^{\pm 1}\rangle$ under the map
  $\Chi_{\theta, \partial}$ is equal to the determinant of the matrix
  corresponding to $L$ \emph{via} the second isomorphism of Corollary
  \ref{cor:azumaya}.  Exactly in the same way, we prove that the image
  of an element $L \in k[x]\langle\partial^{\pm 1}\rangle$ under 
   $\Chi_{x, \partial}$ is equal to the determinant of the matrix
  corresponding to~$L$ \emph{via} the last isomorphism of Corollary~\ref{cor:azumaya}.  Keeping trace of all the identifications, the
  theorem follows. 
\end{proof}


\section{Algorithms}
\label{sec:algo}

This section describes our main algorithm. While the most natural
question is arguably to compute $\Chi_{x,\partial}(L)$ for an element
$L$ of $k[x]\langle \partial^{\pm 1} \rangle$, the formula that gives
$\Chi_{\theta,\partial}(L)$ for $L$ in $k[\theta]\langle \partial^{\pm
  1} \rangle$ of Lemma~\ref{lem:pcurvtheta} leads to a faster
algorithm than its counterpart in $x,\partial$.

As a consequence, we start by discussing conversion algorithms to 
rewrite an operator given in $k[x]\langle \partial^{\pm 1} \rangle$ to 
$k[\theta]\langle \partial^{\pm 1} \rangle$ (\S\ref{ssec:conversions}). 
We continue with algorithms to compute the matrix factorials that arise 
in Lemma~\ref{lem:pcurvtheta} (\S \ref{ssec:factorial}) and with a 
\emph{numerically stable} algorithm to compute the characteristic 
polynomial of a matrix over a ring of power series 
(\S\ref{ssec:charpoly}). Finally, in \S\ref{ssec:mainalgo}, we present 
our main algorithm.

The costs of all our algorithms are given in terms of operations in
$k$. We use standard complexity notation: $\Mpol:\N\to\N$ denotes a
function such that for any ring $A$, polynomials in $A[x]$ of degree
at most $m$ can be multiplied in $\Mpol(m)$ operations in $A$; $\Mpol$
must also satisfy the super-linearity conditions
of~\cite[Chapter~8]{GaGe03}. Using the Cantor-Kaltofen
algorithm~\cite{CaKa91}, one can take $\Mpol(m) = O(m \log(m)
\log\log(m))$.

Let $\omega$ be an exponent such that matrices of size $n$ over a ring 
$A$ can be multiplied in $O(n^\omega)$ operations in $A$; using the 
algorithms of~\cite{CoWi90,Williams12}, we can take $\omega \le 2.38$. 
We assume that $\omega > 2$, so that costs such as $\Mpol(n^2) \log(n)$ 
are negligible compared to~$n^\omega$.

Finally, the soft-O notation $O\tilde{~}(\ )$ indicates the omission
of polylogarithmic factors.


\subsection{Conversion algorithms}
\label{ssec:conversions}

\noindent{\bf\rm From $k[\theta]$ to $k[\theta^p-\theta]$.}
Take $f$ of degree $d$ in $k[\theta]$, and suppose that $f$ lies in
the subring $k[\theta^p-\theta]$ of $k[\theta]$. Thus, it can be
written as $f=\psi(\theta^p-\theta)$, for some $\psi$ in $k[Z]$ of
degree $e-1=d/p$. Our goal is to compute $\psi$.

Consider the power series $t = -Z-Z^p -Z^{p^2} -Z^{p^3}\cdots$ in
$k[[Z]]$; it satisfies the relation $t^p-t = Z$. As a result, in the
power series ring $k[[Z]]$, the equality $f(t) = \psi(Z)$ holds (the
composition $f(t)$ is well-defined, since $t$ has positive valuation).
Thus, to compute $\psi$, it is enough to compute $f(t) \bmod Z^{e}$,
for which only the knowledge of $f \bmod Z^e$ is needed. We will
call such an algorithm {\tt decompose\_central}.

In the common case where $e \le p$, $\psi$ is simply obtained as
$\psi=f(-Z) \bmod Z^{e}$, which is computed in time~$O(e)=O(d/p)$.  In
general, though, we are not able to compute $f(t) \bmod Z^{e}$ in time
quasi-linear in $e$ for the moment; one possible solution is
Bernstein's algorithm, with a running time of $O(p
\Mpol(e)\log(e))=O(\Mpol(d)\log(d))$ operations in
$k$~\cite{Bernstein98b}.

Remark that if $k$ is a finite field, and if we use a {\em boolean}
complexity model (which allows us to lift computations to $\Z$), the
Kedlaya-Umans composition algorithm~\cite{KeUm11} has a running time
almost linear in both $e$ and $\log(|k|)$.

\smallskip\noindent{\bf \rm From $k[x]\langle\partial\rangle$ to
  $k[\theta]$.} Take $f$ in $k[x]\langle \partial \rangle$, of the
form $f=\sum_{i = 0}^d f_i x^i \partial^i$. To rewrite $f$ in
$k[\theta]$, notice as in~\cite{BoSc05,BoChLe08} that this amounts to
multiplying the vector of coefficients of $f$ by the inverse of a
Stirling matrix, which can be done in time $O(\Mpol(d)\log(d))$. We
call this algorithm {\tt x\_d\_to\_theta}.

\smallskip\noindent{\bf \rm From $k[x]\langle\partial^{\pm 1}\rangle$
to $k[\theta]\langle\partial^{\pm 1}\rangle$.} Finally, we describe an
algorithm {\tt x\_d\_to\_theta\_d} that rewrites an operator given in
$k[x]\langle\partial^{\pm 1}\rangle$ on $k[\theta]\langle\partial^{\pm
  1}\rangle$. Take $L$ in $k[x]\langle \partial^{\pm 1}\rangle$, of
the form
$$L = f_{-s}(x) \partial^{-s} + \cdots + f_0(x) + \cdots + f_r(x)
\partial^r,$$ all $f_i$'s being in $k[x]$, of degree at most $d$.
For $i=-s,\dots,r$, let us write $f_i$ as $f_i=\sum_{0 \le j \le d}
f_{i,j} x^j$. Reordering coefficients, we can write 
$f = h_{-s-d} \partial^{-s-d} + \cdots + h_0 + \cdots + h_r
\partial^r,$
with $h_\ell = \sum_{j=0}^d f_{j+\ell,j} x^j \partial^j$ for all $\ell$.
We apply Algorithm {\tt x\_d\_to\_theta} to all $h_\ell$'s, 
allowing us to obtain $f$ as 
$$f = g_{-s-d}(\theta) \partial^{-s-d} + \cdots + g_0(\theta) + \cdots
+ g_r(\theta)\partial^r,$$ for a cost of $O((s+r+d)\Mpol(d)\log(d))$
operations in $k$.


\subsection{Matrix factorials} 
\label{ssec:factorial}

For an $(n \times n)$ matrix $\mB$ in $\mathscr{M}_n(k(\theta))$, and
for an integer $s$, we will denote by ${\rm Fact}(\mB, s)$ the product
$${\rm Fact}(\mB,s) = \mB(\theta) \cdot \mB(\theta+1) \cdots \mB(\theta+s-1).$$

In this paragraph, we describe an algorithm {\tt factorial} that does
the following: given a matrix $\mB$ in
$\mathscr{M}_n(k[\theta])$, with polynomial entries of degree less
than $m$, compute ${\rm Fact}(\mB,s) \bmod \theta^m$. Our main
interest will be in cases where $m \ll s$; our goal is to avoid the
cost linear in $s$ that would follow from computing the product in the
naive manner.

In the special case $n=1$ (so we consider a polynomial $B$ instead of
matrix $\mB$) and $s=p$ (which is the main value we will be interested
in), we are able to obtain a cost logarithmic in $p$. Consider indeed
the bivariate polynomial $P(\theta,\eta)= (\eta^p - \eta) - (\theta^p
- \theta)$. Then, ${\rm Fact}(B,p)$ is the resultant in $\eta$ of
$P(\theta,\eta)$ and $B(\eta)$. This resultant (as well as its
reduction modulo $\theta^m$) can be computed by first reducing
$\eta^p-\eta$ modulo $B$, with a cost polynomial in $\log(p)$. Note in
addition that if we consider $\theta^p - \theta$ instead of $\theta$
as the second variable, this method yields \emph{without any further
  computation} a writing of ${\rm Fact}(B,p)$ as a polynomial in
$\theta^p - \theta$.

Unfortunately, in the case $n >1$, the resultant approach used above
does not apply any longer; as a matter of fact, no solution is known
with cost polynomial in $\log(p)$.

We will rely on an approach pioneered by Strassen~\cite{Strassen76}
and the Chudnovsky's~\cite{ChCh88}, using baby steps/giant steps
techniques. This idea was revisited in~\cite{BoGaSc07}, and led to the
following result~\cite[Lemma~7]{BoClSa05}: provided $p > m$, one can
compute ${\rm Fact}(\mB, p) \bmod \theta^m$ using $\softO(n^\omega
m^{3/2} p^{1/2})$ operations in $k$ (that result is stated over a
finite field; in our case, we use it over $S=k[\theta]/\theta^m$, but
the algorithm still applies).

We present here a variant of these ideas, better adapted to our
context, with a slightly improved cost with respect to~$m$. In what
follows, we call {\tt shift} an algorithm such that ${\tt
  shift}(B,i)=B(\theta+i)$ (we will also use this notation for
matrices of polynomials); Algorithm {\tt shift} can be implemented
using $O(\Mpol(m)\log(m))$ operations in $k$~\cite{GaGe03}, if $\deg(B)\le m$.

\noindent\hrulefill

\noindent {\bf Algorithm} {\tt factorial\_square}

\noindent{\bf Input:} matrix $\mB$, integers $s, m$

\noindent{\bf Output:} ${\rm Fact}(\mB,s^2) \bmod \theta^m$

\smallskip\noindent 1.\ {\bf for} $i=0,\dots,s-1$, compute $\mB_i = {\tt shift}(\mB, i)$

{\sc Cost:} $O(n^2 s \Mpol(m)\log(m))$, since we call {\tt shift} $n^2 s$ times

\smallskip\noindent 2.\ compute $\mC = \mB_0 \cdots \mB_{s-1}$

{\sc Cost:} $O(n^\omega \Mpol(ms) \log(s))$ using~\cite[Algorithm~10.3]{GaGe03}

{\sc Remark:} $\mC = \mB(\theta) \cdot \mB(\theta+1) \cdots \mB(\theta+s-1)$

\smallskip\noindent 3.\ {\bf for} $i=0,\dots,s-1$, compute $\mC_i = \mC \bmod (\theta-si)^m$

{\sc Cost:} $O(n^2 \Mpol(ms) \log(s))$ using~\cite[Corollary~10.17]{GaGe03}

\smallskip\noindent 4.\ {\bf for} $i=0,\dots,s-1$, compute $\mD_i = {\tt shift}(\mC_i,si)$

{\sc Cost:} $O(n^2 s \Mpol(m)\log(m))$

{\sc Remark:} $\mD_i$ is also equal to $\mC(\theta+si) \bmod \theta^m$

\smallskip\noindent 5.\ {\bf return} $\mD_{0} \cdots \mD_{s-1} \bmod \theta^m$

{\sc Cost:} $O(n^\omega s \Mpol(m))$

\vspace{-1ex}\noindent\hrulefill

In view of the remarks made in the algorithm, we see that Algorithm
{\tt factorial\_square} computes ${\rm Fact}(\mB,s^2) \bmod \theta^m$
using $O(n^\omega \Mpol(ms) \log(ms))$ operations in $k$. 

This algorithm only deals with product lengths that are perfect squares.
In the general case, we will rely on the following (obvious) equality,
that holds for any integers $s, t$:
$${\rm Fact}(\mB,s+t) = {\rm Fact}(\mB,s) \cdot {\rm Fact}(\mB(\theta+s), t).$$ 
For an arbitrary $s$, this allows us to compute ${\rm Fact}(\mB,s) \bmod \theta^m$
using the base 4 decomposition of $s$ as follows.

\noindent\hrulefill

\noindent {\bf Algorithm} {\tt factorial}

\noindent{\bf Input:} matrix $\mB$, integer $s, m$.

\noindent{\bf Output:} ${\rm Fact}(\mB,s) \bmod \theta^m$

\smallskip\noindent 1.\ Write $s$ in base 4 as $s=\sum_{0 \le i \le N} 4^{e_i}$

{\sc Cost:} no operation in $k$

{\sc Remark:} $N = O(\log(s))$ and $e_i=O(\log(s))$ for all $i$

\smallskip\noindent 2.\ {\bf for} $i=0,\dots,N$, compute $\mB_i = {\tt shift}(\mB, \sum_{0 \le j <i }4^{e_j})$

{\sc Cost:} $O(n^2 \log(s) \Mpol(m)\log(m))$

\smallskip\noindent 3.\ {\bf for} $i=0,\dots,N$, let $\mC_i = {\tt factorial\_square}(\mB_i, 2^{e_i})$

{\sc Cost:} $O(n^\omega \Mpol(m s^{1/2}) \log(ms))$ 

\smallskip\noindent 4.\ {\bf return} $\mC_{0} \cdots \mC_{N}$

{\sc Cost:} $O(n^\omega \Mpol(m) \log(s))$

\vspace{-1ex}\noindent\hrulefill

\begin{lem}\label{lemma:any}
 ~Algorithm {\tt factorial} computes ${\rm Fact}(\mB,s)$
 modulo $\theta^m$ in $O(n^\omega \Mpol(ms^{1/2})\log(ms))$ operations in~$k$.
\end{lem}
\begin{proof}
  Correctness follows from the remarks made prior to the algorithm.
  We claim that the cost given in the lemma is an upper bound on the
  costs of all steps. This is clear for Steps 2 and 4; the only point
  that requires proof is the claim that the overall cost of Step 3 is
  $O(n^\omega \Mpol(m s^{1/2}) \log(ms))$.

  For a given index $i$ in $\{0,\dots,N\}$, the cost incurred by
  calling ${\tt factorial\_square}(\mB_i, 2^{e_i})$ is $O(n^\omega
  \Mpol(m 2^{e_i}) \log(m 2^{e_i}))$, which is $O(n^\omega \Mpol(m
  2^{e_i}) \log(m s))$. Using the super-linearity of $\Mpol$, and the
  fact that $\sum_i 2^{e_i}=O(s^{1/2})$, the total cost is thus
  $O(n^\omega \Mpol(m s^{1/2}) \log(m s))$.
\end{proof}


\subsection{Characteristic polynomials}
\label{ssec:charpoly}

Let $M$ be a square matrix of size $r$ defined over the field of Laurent 
series $k((Z))$. We assume that there exists two nonnegative integers 
$N$ and $v$ such that:
\begin{enumerate}[(a)]
\item all coefficients of $M$ are known at precision $O(Z^N)$;
\item any minor (of any size) of $M$ has $Z$-adic valuation $\geq -v$.
\end{enumerate}
We are going to describe a \emph{numerically stable} algorithm to
compute (a good approximation of) the characteristic polynomial 
$\chi(M) \in k((Z))[X]$ of $M$. 

To do this, we use a rather naive approach: we work in the quotient ring 
$k((Z))[X]/(X^{r+1} - Z)$ which turns out to be isomorphic to $k((X))$,
we compute an ``approximate Hermite form'' of $(X\cdot \text{Id} -
M)$ and then multiply all diagonal coefficients of it to recover the
image in $k((X))$ of the characteristic polynomial of $M$. Because 
$\chi(M)$ has degree $r$, the knowledge of its image in $k((X))$ is
enough to recover it entirely. Let us now precise what we mean by
an \emph{approximate Hermite form}; it is a factorization:
\begin{equation}
\label{eq:Hermite}
X\cdot \text{Id} - M = P \cdot H
\end{equation}
where $P$ is a \emph{unimodular matrix} with coefficients in $k[[X]]$ 
and $H$ is lower triangular modulo $Z^N$.

\noindent\hrulefill

\noindent {\bf Algorithm} {\tt charpoly}

\noindent{\bf Input:}  $M \in \mathscr{M}_r(k((Z)))$ 
and $N, v \in \N$ such that (a), (b)

\noindent{\bf Output:} $\chi(M)$ at precision $O(Z^{N-v})$

\smallskip\noindent 1.\ Compute $M_X = X \cdot \text{Id} - M \in \mathscr{M}_r(k((X)))$

{\sc Cost:} no operation in $k$

\smallskip\noindent 2.\ Compute an approximation Hermite form $(P,H)$ of
$M_X$

{\sc Cost:} $O(r^\omega \Mpol(r(N+v)))$ using procedure {\tt LV} of 
\cite[\S 2.1.5]{Ca12}

{\sc Remark:} all entries of $H$ are known at precision $O(Z^N)$

\smallskip\noindent 3.\ Compute $\chi = \lambda_1 \cdots \lambda_r + O(Z^{N-v})$,

where the $\lambda_i$'s are the diagonal entries of $H$

{\sc Cost:} $O(r \Mpol(rN))$

{\sc Remark:} We shall prove that $\chi = \det(H) = \det(M_X)$.

\smallskip\noindent 4.\ Reorder coefficients of $\chi$ to get $\chi(M)$

{\sc Cost:} no operation in $k$

\smallskip\noindent 5.\ {\bf return} $\chi(M)$

\vspace{-1ex}\noindent\hrulefill

\begin{lem}
  Algorithm {\tt charpoly} outputs $\chi(M)$ at precision $O(Z^{N-v})$ 
  in $O(r^\omega \Mpol(r(N+v)))$ operations in $k$.
\end{lem}
\begin{proof}
  We are going to check the following three items:
 \emph{(i)} the product
  $\lambda_1 \cdots \lambda_r$ is known with precision $O(Z^{N-v})$;
  \emph{(ii)} it can be computed with the announced complexity and \emph{(iii)} we have
  $\chi \equiv \det(H) = \det(M_X) \pmod {Z^{N-v}}$.

  From Eq.~\eqref{eq:Hermite}, we deduce immediately that $M_X$ and
  $H$ share the same determinant. Moreover, from our assumptions, we
  deduce that all minors of $H$ have $Z$-adic valuation $\geq
  -v$. Denoting by $v_Z$ the $Z$-adic valuation, we deduce that
  $$v_Z(\lambda_1 \cdots \lambda_{i-1} \lambda_{i+1} \cdots \lambda_r) 
  \geq -v$$
  for all $i$.
  Setting $\delta = v_Z(\lambda_1 \cdots \lambda_r)$, we get 
  $v_Z(\lambda_i) \leq \delta+v$. Hence $\lambda_i$ is known with
  relative precision at least $N - \delta - v$. (We recall that the 
  relative precision is the difference between the absolute precision and 
  the valuation.) Therefore the product $\lambda_1 \cdots \lambda_r$
  is known with relative precision $N - \delta - v$. Since it has 
  valuation $\delta$, it is known with absolute precision
  $O(Z^{N-v})$. This gives \emph{(i)}. \emph{(ii)} follows similarly from the lower 
  bound on the valuation on the $\lambda_i$'s.
  
  Finally, to prove \emph{(iii)}, we remark that if $A$ and $B$ are two matrices 
  such that $B - A$ has only one nonzero coefficient $a$ located in 
  position $(i,j)$, then all minors of $B$ differ from the corresponding 
  minor of $A$ by either $0$ or the product of $a$ by another minor of 
  $A$. Using this, we can clear one by one all entries of $H$ lying above 
  the diagonal without changing the value of the determinant modulo 
  $Z^{N-v}$.
\end{proof}


\subsection{The main algorithm}
\label{ssec:mainalgo}

We can now give our main algorithm to compute the mappings
$\Chi_{\theta,\partial}$ and $\Chi_{x,\partial}$. We start with the
former, which is computed by means of matrix factorials.
The central operation is to compute $\Chi_{\theta,\partial}(L)$ for
some $L$ in $k[\theta]\langle \partial \rangle$, of degree $r$ in
$\partial$. For such an operator $L$, we have by definition:
$$\Chi_{\theta,\partial}(L) = {\rm Fact}(g_r, p)\cdot \chi( {\rm
  Fact}(\mB,p))(\partial^p),$$ where as before, $g_r \in k[\theta]$ is
the leading coefficient of $L$ with respect to $\partial$ and $\mB$ is
the companion matrix of $L$.  If $d$ is the maximal degree of the
coefficients of $L$, we know by Lemmas~\ref{lemma:central} and
\ref{lem:norm} that $\Chi_{\theta,
  \partial}(L)=C(\theta^p-\theta,\partial^p)$, where $C\in k[U,V]$ has
degree at most $d$ in $U$ and exactly $r$ in $V$. Our algorithm
computes this polynomial.

We set $\bbeta = {\rm Fact}(\mB,p)$. It is a matrix with coefficients
in $k(\theta)$ but we view it as a matrix over $k((\theta))$
\emph{via} the natural embedding $k(\theta) \hookrightarrow
k((\theta))$. Let also $v$ denote the number of roots (counted with
multiplicity) of $g_r$ in the prime field $\F_p$. We have $v \leq d$;
besides, $v$ equals the $\theta$-adic valuation of $\gamma= {\rm
  Fact}(g_r, p)$, seen as an element of $k[[\theta]]$.

\begin{lem}
All minors of $\bbeta$ have $\theta$-adic valuation $\geq -v$.
\end{lem}
\begin{proof}
If $M$ is a matrix, we denote by $\Lambda^i M$ its matrix of minors
of size $i$. From the definition of $\bbeta$, we get
$\Lambda^i \bbeta = {\rm Fact}(\Lambda^i \mB,p)$
for all $i$. Now remark that $g_r$ is a common denominator for all
the entries of $\Lambda^i \mB$. Hence the matrix $\gamma \cdot \Lambda^i
\bbeta$ has coefficients in $k[\theta] \subset k[[\theta]]$ and we are
done.
\end{proof}

Before giving our algorithm, we mention another subroutine, {\tt
  count\_roots}, which returns the number of roots in $\F_p$ of a
polynomial $g$ of degree $d$ in $k[\theta]$, counted with
multiplicities. By computing the squarefree decomposition of $g$, and
estimating the degree of the gcd of each factor with
$\theta^p-\theta$, this can be done in $O(\Mpol(d)\log(dp))$
operations in $k$.

\noindent\hrulefill

\noindent {\bf Algorithm} {\tt Xi\_theta\_d}

\noindent{\bf Input:} operator $L$ in $k[\theta]\langle \partial \rangle$

\noindent{\bf Output:} $C \in k[U,V]$ such that
$\Chi_{\theta,\partial}(L) = C(\theta^p - \theta, \partial^p)$

\smallskip\noindent 1.\ let $g_r$ be the leading coefficient of $L$ in
$\partial$,  $\mB$ be the 

companion matrix of $L$ and $\mB^\star= g_r \mB$

{\sc Cost:} no operation in $k$

\smallskip\noindent 2.\ compute $v = \text{\tt count\_roots}(g_r)$

{\sc Cost:} $O(\Mpol(d)\log(dp))$

\smallskip\noindent 3.\ compute $\gamma = {\tt factorial}(g_r,\: p,\: d + 2v + 1)$

{\sc Cost:} $O(\Mpol(d p^{1/2})\log(dp))$ using 
Lemma~\ref{lemma:any}

{\sc Remark:} A better complexity is possible using resultants

\smallskip\noindent 4.\ compute $\bbeta^\star = {\tt factorial}(\mB^\star,\: p,\: d + v + 1)$

{\sc Cost:} $O(r^\omega \Mpol(d p^{1/2})\log(dp))$ using Lemma~\ref{lemma:any}

\smallskip\noindent 5.\ compute $\bbeta = \gamma^{-1} \bbeta^\star \in
\mathscr M_r(k((\theta)))$ at precision $O(\theta^{d+1})$.

{\sc Cost:} $O(r^2 \Mpol(d))$

\smallskip\noindent 6.\ compute $\chi = \gamma \cdot {\tt charpoly}(\bbeta,\: d+1,\: v)$

{\sc Cost:} $O(r^\omega \Mpol(dr))$

{\sc Remark:} $\chi$ is in $k[[\theta]]$ and is known at precision $O(\theta^{d+1})$.

\smallskip\noindent 7.\ {\bf for} $i=0,\dots,r$,

compute $C_i = {\tt decompose\_central}({\tt coeff}(\chi, X^i))$

{\sc Cost:} $O(r d)$ if $d \le p$, $O(r \Mpol(dp)\log(dp))$ if $d \ge p$

\smallskip\noindent 8.\ {\bf return} $\sum_{i=0}^r C_i(U) V^i$

\vspace{-1ex}\noindent\hrulefill

\begin{prop}
\label{prop:ChiThetaPartial}
  Algorithm {\tt Xi\_theta\_d} is correct and, provided that
  $p \geq d$, runs in time
  $$O(r^\omega \Mpol(d p^{1/2})\log(dp) + 
  r^\omega \Mpol(rd))=\softO(r^\omega d p^{1/2} + r^{\omega+1} d).$$
\end{prop}
\begin{proof}
It remains only to prove that the matrix $\bbeta$ of Step 5 can be
computed at precision $O(\theta^{d+1})$ in the given
complexity. Remark that $\gamma^{-1}$ is known at precision
$O(\theta^{d+1})$ and has valuation $-v$. Since $\bbeta^\star$ has
nonnegative valuation and is known at precision $O(\theta^{v+d+1})$,
the result follows.
\end{proof}

\begin{figure*}\centering
{\scriptsize 

\begin{tabular}{|ll||>{\raggedleft}p{4.5em}|>{\raggedleft}p{4.5em}|>{\raggedleft}p{4.5em}|>{\raggedleft}p{4.5em}|>{\raggedleft}p{4.5em}|>{\raggedleft}p{4.5em}|>{\raggedleft}p{4.5em}|l}
\cline{3-9}
 \multicolumn{2}{c|}{} & \multicolumn{7}{c|}{$\mathbf{p}$} \cr
\cline{3-9}
 \multicolumn{2}{c|}{} & \centering $\mathbf{83}$ & \centering $\mathbf{281}$ & \centering $\mathbf{983}$ & \centering $\mathbf{3\:433}$ & \centering $\mathbf{12\:007}$ & \centering $\mathbf{42\:013}$ & \centering $\mathbf{120\:011}$ & \cr
\cline{1-9}
 $d = 5$, & $r = 5$ & $0.11$~s & $0.26$~s & $0.75$~s & $1.95$~s & $5.09$~s & $12.43$~s & $33.78$~s \cr
 $d = 5$, & $r = 8$ & $0.19$~s & $0.47$~s & $1.32$~s & $3.43$~s & $9.20$~s & $22.55$~s & $65.25$~s \cr
 $d = 5$, & $r = 11$ & $0.26$~s & $0.66$~s & $1.85$~s & $5.01$~s & $14.68$~s & $37.91$~s & $104.86$~s \cr
 $d = 5$, & $r = 14$ & $0.37$~s & $0.86$~s & $2.38$~s & $6.61$~s & $20.52$~s & $59.47$~s & $154.76$~s \cr
 $d = 5$, & $r = 17$ & $0.52$~s & $1.21$~s & $3.26$~s & $8.29$~s & $24.18$~s & $76.81$~s & $234.28$~s \cr
 $d = 5$, & $r = 20$ & $0.76$~s & $1.74$~s & $4.67$~s & $11.93$~s & $33.88$~s & $109.02$~s & $298.72$~s \cr
 $d = 8$, & $r = 20$ & $1.12$~s & $2.41$~s & $6.69$~s & $18.86$~s & $56.24$~s & $239.49$~s & $881.45$~s \cr
 $d = 11$, & $r = 20$ & $1.96$~s & $4.33$~s & $10.42$~s & $30.87$~s & $92.84$~s & $388.50$~s & $922.34$~s \cr
 $d = 14$, & $r = 20$ & $3.05$~s & $6.11$~s & $14.45$~s & $45.53$~s & $141.81$~s & $507.89$~s & $1\:224.98$~s \cr
 $d = 17$, & $r = 20$ & $5.26$~s & $9.19$~s & $20.85$~s & $56.83$~s & $195.74$~s & $699.08$~s & $1\:996.87$~s \cr
 $d = 20$, & $r = 20$ & $7.76$~s & $13.94$~s & $28.40$~s & $82.43$~s & $240.47$~s & $889.48$~s & $2\:419.56$~s \cr
\cline{1-9}
\end{tabular}
}
\caption{Average running time on random inputs of various sizes}\label{fig:table}
\end{figure*}

Finally, we give an algorithm that computes $\Chi_{x,\partial}(L)$,
for $L$ in $k[x]\langle \partial \rangle$. Since
$\Chi_{x,\partial}(L)$ is a polynomial in $x^p$ and $\partial^p$, the
output will be a polynomial $D$ in $k[U,V]$ such that
$D(x^p,\partial^p)=\Chi_{x,\partial}(L)$.  We let $d$ and $r$ 
be the degrees of $L$ in respectively $x$ and $\partial$.

\noindent\hrulefill

\noindent {\bf Algorithm} {\tt Xi\_x\_d}

\noindent{\bf Input:} operator $L$ in $k[x]\langle \partial \rangle$

\noindent{\bf Output:} $C \in k[U,V]$ such that
$\Chi_{x,\partial}(L) = C(x^p, \partial^p)$

\smallskip\noindent 1.\ compute $L'=\text{{\tt x\_d\_to\_theta\_d}}(L)$

{\sc Cost:} $O((r+d)\Mpol(d)\log(d))$

{\sc Remark:} $L'$ has the form $g_{-d}(\theta) \partial^{-d} + \cdots + g_r(\theta) \partial^r$

\smallskip\noindent 2.\ compute $C=\text{{\tt Xi\_theta\_d}}(L' \partial^d) \in k[U,V]$

{\sc Cost:} 

$O((r+d)^\omega \Mpol(d p^{1/2})\log(dp) + (r+d)^\omega \Mpol((r+d) d))$

{\sc Remark:} This complexity is correct even if $p < d$

\smallskip\noindent 3.\ {\bf return} $C(UV, V)/V^d$

{\sc Cost:} no operation in $k$

\vspace{-1ex}\noindent\hrulefill

\begin{theo}
\label{theo:ChiXPartial}
  Algorithm {\tt Xi\_x\_d} is correct and runs in time
  $$O((r+d)^\omega \Mpol(d p^{1/2})\log(dp) + 
  (r+d)^\omega \Mpol((r+d)d))$$
  which is $\softO((r+d)^\omega d p^{1/2} + (r+d)^{\omega+1} d)$.
\end{theo}

\begin{proof}
Clear from what precedes.
\end{proof}

We can use Algorithm {\tt Xi\_x\_d} to compute $\Chi_{x,
  \partial}(L)$ for any $L \in k(x)\langle \partial \rangle$. Indeed,
we can write such an $L$ as $f(x) \: L_0$ with $f(x) \in k(x)$ and
$L_0 \in k[X] \langle \partial \rangle$. Now we can compute $\Chi_{x,
  \partial}(L_0)$ using Algorithm {\tt Xi\_x\_d} and finally recover
$\Chi_{x, \partial}(L)$ just by multiplying $\Chi_{x, \partial}(L_0)$
by $f(x)^p$.

We conclude this section by a final remark concerning Fourier
transform.  Recall that $k[x]\langle \partial \rangle$ is endowed by a
ring automorphism defined by $x \mapsto -\partial$, $\partial \mapsto
x$. It is the so-called \emph{Fourier transform}. If $L$ is some
differential operator of degrees $(d,r)$ in $(x,\partial)$, its
Fourier transform $\hat L$ has degrees $(r,d)$ in
$(x,\partial)$. Moreover, using an analogue for $k[x]\langle
\partial^{\pm 1} \rangle$ of Proposition \ref{prop:norm}, one can
check that $\Chi_{x, \partial}$ commutes with Fourier transform. As a
consequence, if we want to compute $\Chi_{x, \partial}(L)$ for a
differential operator $L$ of degrees $(d,r)$ in $(x,\partial)$, with
$d \geq r$, instead of using directly Algorithm {\tt Xi\_x\_d}, it is
more clever to compute the inverse Fourier transform of
$\Chi_{x,\partial}(\hat L)$.

Applying the Fourier transform or its inverse requires only $\softO(dr)$ 
operations in $k$, so the whole computation is dominated by the cost of 
computing $\Chi_{x,\partial}(\hat L)$, which is
  $$ 
  O((r+d)^\omega \Mpol(r p^{1/2})\log(rp) + 
  (r+d)^\omega \Mpol((r+d)r)). 
  $$
This is better than the complexity announced in
Theorem~\ref{theo:ChiXPartial} when $d \geq r$.  Using the fact that
the $p$-curvature of $L$ is nilpotent if and only if $\Chi_{x,
  \partial}(L)$ is a product of an element in $k[x]$ by
$\partial^{pr}$, we deduce the following.

\begin{cor}
  There exists an algorithm that decides whether a differential
  operator $L \in k[x]\langle\partial\rangle$ of degrees $(d,r)$ in
  $(x,\partial)$ has nilpotent $p$-curvature in time
  $$\softO((r+d)^\omega \min(d,r) \:p^{1/2} + (r+d)^{\omega+1} \min(d,r)).$$
\end{cor}


\section{Implementation and timings}
\label{sec:implementation}

We implemented our algorithms in Magma; the source code is available
at \url{https://github.com/schost}. Figure~\ref{fig:table} gives
running times for random operators of degrees $(d,r)$ in $k[x]\langle
\partial \rangle$, obtained with Magma V2.19-4 on an AMD Opteron 6272
machine with 4 cores at 2GHz and 8GB RAM, running Linux.  Very large
values of $p$ are now reachable; timings do not quite reflect the
predicted behavior with respect to $p$, for reasons unknown to us
(experiments on other machines gave similar results). For the largest
examples, the bottleneck is actually memory: the {\tt factorial}
algorithm of Subsection~\ref{ssec:factorial} requires to store
$O(p^{1/2})$ matrices.



Using our implementation, we have computed characteristic polynomials
of $p$-curvatures for some linear differential operators with physical
relevance.  These operators annihilate multiple parametrized integrals
of algebraic functions occurring in the study of the susceptibility of
the square lattice Ising model. We considered the operator
$\phi_H^{(5)}$ of~\cite[Appendix~B.3]{BoHaMaZe07}: it belongs to
$(\mathbb{Z}/27449\, \mathbb{Z})[x]\langle \partial \rangle$, has
degree 28 in $\partial$ and $108$ in $x$.  We found that the
characteristic polynomial of its $27449$-curvature is equal to
$C(x^{27449}, V)$, where $C(U,V)$ is a polynomial 
of degree $(108, 28)$ and valuation $(17,17)$ in $(U,V)$. 
This high valuation is in agreement with the empirical prediction that the (globally 
nilpotent) minimal-order operator for $\phi_H^{(5)}$ has order 17.

We also considered a right-multiple, of degree 77 in $\partial$ and
$140$ in $x$, of the operator $L_{23}$ mentioned
in~\cite[\S4.3]{BoHaJeMaZe10}, and we computed the characteristic
polynomial of its $p$-curvature for $p\in \{32647,32713\}$.
Note that for all these operators, $p$-curvatures themselves are impossible to compute using current algorithms.

{\tiny
\bibliographystyle{abbrv}

} 

\end{document}